\documentclass[12 pt]{article}
\usepackage{cite,amsmath,amsfonts,amsthm,fullpage}
\usepackage{youngtab}
\usepackage{cite,amsmath,amsfonts,amsthm,fullpage, amssymb, mathtools, url}
\usepackage[usenames,svgnames]{xcolor}
\usepackage{youngtab, enumitem, comment, mathdots, graphicx, float}
\usepackage[symbol]{footmisc}
\usepackage[pagebackref=false]{hyperref} 

\theoremstyle{plain}
\theoremstyle{definition}
\newtheorem{theorem}{Theorem}[section]
\newtheorem{lemma}[theorem]{Lemma}
\newtheorem{definition-theorem}[theorem]{Definition-Theorem}
\newtheorem{definition-proposition}[theorem]{Definition-Proposition}

\newtheorem{corollary}[theorem]{Corollary}
\newtheorem{example}{Example}[section]
\newtheorem{examples}{Example}[subsection]

\newtheorem{remark}{Remark}[section]
\newtheorem{remarks}{Remarks}[section]
\newtheorem{definition}{Definition}[section]

\numberwithin{equation}{section} % requires package amsthm

\DeclareMathOperator{\End}{End}

\DeclareMathOperator{\sgn}{sgn}

\def\ra{{\rightarrow}}

\def\det{\mathrm {det}}

\def\Pf{\mathrm {Pf}}

\def\End{\mathrm {End}}

\def\&{&{\hskip -20pt}}

\def\be{\begin{equation}}
\def\ee{\end{equation}}

\def\bea{\begin{eqnarray}}
\def\eea{\end{eqnarray}}

\def\bt{\begin{theorem}}
\def\et{\end{theorem}}

\def\bex{\begin{example}\small \rm}
\def\eex{\end{example}}

\def\bexs{\begin{examples}\small \rm}
\def\eexs{\end{examples}}

\def\ra{\rightarrow}
\def\ss{\subset}

\def\br{\begin{remark}\small \rm}
\def\er{\end{remark}}

\def\CC{{\mathcal C}}

\def\FF {{\mathcal F}}

\def\HH{{\mathcal H}}

\def\PP{{\mathcal P}}

\def\Mb{\mathbf{M}}
\def\Nb{\mathbf{N}}

\def\Zb{\mathbf{Z}}

\def\tb{\mathbf{t}}

\def\xb{\mathbf{x}}

\newcommand{\DP}{\mathop\mathrm{DP}\nolimits}

\def\bp{\begin{Proposition}\rm}
\def\ep{\end{Proposition}}
\def\bc{\begin{corollary}}
\def\ec{\end{corollary}}
\def\bl{\begin{lemma}\em}
\def\el{\end{lemma}}
\def\be{\begin{equation}}
\def\ee{\end{equation}}
\def\br{\begin{remark}\rm\small}
\def\er{\end{remark}}
\def\brs{\begin{remarks}.\\ \rm\
\begin{enumerate}}
\def\ers{\end{enumerate}\end{remarks}}
\def\bea{\begin{eqnarray}}
\def\eea{\end{eqnarray}}

\begin{document}

\begin{center}
\begin{Large}\fontfamily{cmss}
\fontsize{17pt}{27pt}
\selectfont
	\textbf{Bilinear expansion of Schur functions in Schur Q-functions: a fermionic approach}
	\end{Large}
	
\bigskip \bigskip
\begin{large} 
J. Harnad$^{1, 2}$\footnote[1]{e-mail:harnad@crm.umontreal.ca} 
and A. Yu. Orlov$^{3}$\footnote[2]{e-mail:orlovs55@mail.ru} 
 \end{large}
 \\
\bigskip

\begin{small}
$^{1}${\em Centre de recherches math\'ematiques, Universit\'e de Montr\'eal, \\C.~P.~6128, succ. centre ville, Montr\'eal, QC H3C 3J7  Canada}\\
$^{2}${\em Department of Mathematics and Statistics, Concordia University\\ 1455 de Maisonneuve Blvd.~W.~Montreal, QC H3G 1M8  Canada}\\
$^{3}${\em Shirshov Institute of Oceanology, Russian Academy of Science, Nahimovskii Prospekt 36, Moscow 117997, Russia }
\end{small}
 \end{center}
\medskip

\begin{abstract}
An identity is derived expressing Schur functions as sums over products of pairs of Schur $Q$-functions,
generalizing previously known special cases.  This is shown to follow from their representations as vacuum 
expectation values (VEV's) of products of either charged or neutral fermionic creation and annihilation  operators, Wick's theorem
and a factorization identity for VEV's of products of two mutually anticommuting sets of neutral fermionic operators.
 \end{abstract}

\bigskip

%%%%%%%%%%%%%%%%%%%%%%%%%%%  Section 1. Introduction %%%%%%%%%%%%%%%%%%%%%%%%%%
\section{Introduction}
\label{introduction}

 Fermionic methods are central to Sato's construction of $\tau$-functions for the KP infinite integrable 
 hierarchy,  as well as the BKP hierarchy \cite{Sa, DJKM1, DJKM2, JM, KvdL, vdLO, HB}. 
 In this work, we make use of the relations between charged and neutral fermionic operators 
 to derive a bilinear identity relating Schur functions \cite{Mac}, which are the basic building blocks 
 for solutions of the KP hierarchy \cite{Sa, JM}, to Schur's  $Q$-functions, which play a similar r\^ole 
 with respect to the BKP hierarchy \cite{DJKM1, You}.
 
An identity was derived  in \cite{BHH},  expressing determinants of  submatrices of skew matrices as  sums 
over products of the Pfaffians of their principal minors. Geometrically, this may be interpreted as a 
bilinear relation between the {\em Pl\"ucker map}, which embeds Grassmannians  of $k$-dimensional subspaces of a given vector space into the projectivization  of the space of exterior $k$-forms, and the {\em Cartan map} \cite{Ca}, which  embeds the Grassmannian   of maximal isotropic subspaces with respect to a complex scalar product into the projectivization of the irreducible spinor modules.
 The result  in \cite{BHH} was  based on Cartan's construction of bilinear forms on Clifford modules, with values in spaces of  homogeneous exterior forms. The  main result derived here is Theorem \ref{S_QQ}, Section \ref{schur_Q_schur_bilin}), which  may be viewed  as a function theoretic realization of this identity,  with the determinant identified, through the Jacobi-Trudi identity \cite{Mac} with  the Schur function, and the Pfaffians, with Schur's  $Q$-functions.
 
 Section \ref{fermionic_fock} recalls the definition of creation and annihilation operators on fermionic Fock space
  as linear generators of an infinite dimensional Clifford algebra. Two mutually 
 anticommuting subalgebras generated by neutral fermions are defined, and a  key factorization Lemma \ref{factorization_lemma}  
  given for vacuum state expectation values (VEV's) of products of linear elements. The representation of Schur functions and 
 Schur Q-functions as VEV's of products of creation and annihilation operators conjugated by elements of the 
 infinite abelian groups that generate KP and BKP flows is recalled in Section \ref{schur_schur_Q_fermions}.
 Section \ref{polariz_binary} introduces the notion of {\em polarizations} associated to integer partitions,
 and the associated products of neutral fermion operators determined by binary sequences.
The main result is Theorem \ref{S_QQ}, Section \ref{schur_Q_schur_bilin}, which gives an expression for Schur functions, 
 evaluated on the odd KP flow variables, as  sums over products of pairs of Schur  $Q$-functions, 
generalizing certain previously known special cases \cite{Mac} to arbitrary Schur functions.

%%%%%%%%%%%%%%%%%  Section 2. Fermionic Fock space  %%%%%%%%%%%%%%%%%%%%%%%%%%
\section{Fermionic Fock space}
\label{fermionic_fock}

%%%%%%%%%%%%%%%%%  Subsection 2. 1 Charged and neutral fermions   %%%%%%%%%%%%%%%%%%%%%%%%%%
\subsection{Charged and neutral fermions }
\label{fermions}

For a separable Hilbert space $\HH$,  with orthonormal basis $\{e_j\}_{j\in \Zb}$,
the corresponding fermionic Fock space $\FF$ is defined as the semi-infinite wedge product  space \cite{Sa, JM}:
\be
\FF = \Lambda^{\infty/2} \HH = \bigoplus_{n\in \Zb} \FF_n,
\ee
with elements  denoted as {\em ket} vectors  $| w\rangle$.
The  sector $\FF_n$  with fermionic charge $n\in \Zb$ has an orthonormal basis, denoted
$\{|\lambda;n\rangle\}$, labelled by pairs $(\lambda, n)$  of an 
integer partitition $\lambda$ and the integer $n$,  defined by:
\be
|\lambda;n\rangle := e_{l_1} \wedge e_{l_2} \wedge \cdots , 
\ee
where the infinite sequence of  integers  $\{l_i\}_{i\in \Nb^+}$, called {\em particle positions}, 
 are related to the parts $(\lambda_1\ge \lambda_2 \cdots , \lambda_{\ell(\lambda)}, 0, \dots)$ 
 of  $\lambda$ by
\be
l_i := \lambda_i - i +n , \quad i\in \Nb^+.
\ee
The {\em length} $\ell(\lambda)$ of the partition $\lambda$ is the number of positive parts  $\{\lambda_i\}_{i=1, \dots, \ell(\lambda)}$ 
and the finite sequence is completed by adding an infinite sequence of $0$'s following these. Its {\em weight} is
\be
|\lambda| := \sum_{i=1}^{\ell(\lambda)} \lambda_i.
\ee
The particle positions  $\{l_i\}_{i\in \Nb^+}$, form a strictly decreasing sequence that saturates, after $\ell(\lambda)$ terms,
to all subsequent consecutively decreasing integers.

The algebra of fermionic operators on $\FF$ form an irreducible representation 
\bea
\Gamma: \CC_{\HH+\HH^*}(Q)&\&\ra \End(\FF)\cr
\Gamma: \xi &\& \mapsto \Gamma_\xi
\eea
 of the Clifford algebra  $\CC_{\HH+\HH^*}(Q)$ on $\HH + \HH^*$ corresponding to the (complex) scalar product $Q$ defined by
\be
Q(u+\mu, v + \nu) := \mu(v) + \nu(u), \quad u, v \in \HH, \ \mu, \nu \in \HH^*.
\ee
 These are realized as endomorphisms of $\FF$, with the linear elements  acting  by exterior and interior multiplication:
 \be
 \Gamma_v  = v \wedge  , \quad \Gamma_\mu   = i_\mu  , \quad v \in \HH, \ \mu\in \HH^*.
 \ee
 
 Denoting the dual basis for $\HH^*$ as $\{e^j\}_{j\in \Zb}$,   with
\be
e^j(e_k) = \delta_{jk},
\ee
 $\CC_{\HH+\HH^*}(Q)$ is generated by the scalars and linear elements, with the
 representations of the basis elements  denoted
\be
\psi_j := e_j \wedge,  \quad \psi^\dag_j := i_{e^j}, \quad j \in \Zb.
\ee
These are referred to as (charged) fermionic creation and annihilation operators, respectively,
and  satisfy the anticommutation  relations:
\be
[\psi_j,\psi_k]_+= [\psi^\dag_j,\psi^\dag_k]_+=0,\quad [\psi_j,\psi^\dag_k]_+=\delta_{jk} .
\label{charged-canonical}
\ee

The {\em vacuum} element $|n\rangle$ in each charge sector $\FF_n$ is the basis element
corresponding to the trivial partition $\lambda = \emptyset$:
\be
| n\rangle :=|\emptyset; n \rangle = e_{n-1} \wedge e_{n-2} \wedge \cdots .
\ee
Elements of the dual space $\FF^*$ are denoted as {\em bra} vectors $\langle w |$,
with the dual basis $\{\langle \lambda ;n|\}$ for $\FF^*_n$ defined by the pairing
\be
\langle \lambda; n | \mu; m\rangle = \delta_{\lambda \mu} \delta_{nm}.
\ee
For KP $\tau$-functions, we need only  consider the $n=0$ charge sector $\FF_0$,
and generally drop the charge $n$ symbol, denoting the basis elements simply as
\be
|\lambda\rangle :=|\lambda;0\rangle.
\ee 
 For $j>0$, $\psi_{-j}$ and $\psi^\dag_{j-1}$
(resp. $\psi^\dag_{-j}$ and $\psi_{j-1}$) annihilate the right (resp. left) vacua:
\bea
\psi_{-j} |0\rangle &\&=0, \quad \psi^\dag_{j-1} |0\rangle =0, \quad \forall  j >0,
\label{vac_annihil_psi_j_r}
 \\
\langle 0| \psi^\dag_{-j}  &\&=0, \quad \langle 0 | \psi_{j-1}  =0, \quad \forall  j >0.
\label{vac_annihil_psi_j_l}
\eea

 Neutral fermions $\phi^+_j$ and $\phi^-_j$ are defined \cite{DJKM1} as 
\be
\phi^+_j :=\frac{\psi_j+ (-1)^j\psi^\dag_{-j}}{\sqrt 2},\quad 
\phi^-_j :=i\frac{\psi_j-(-1)^j \psi^\dag_{-j}}{\sqrt 2},\quad j \in \Zb
\label{charged-neutral}
\ee
(where $i=\sqrt{-1}$), and satisfy
\be
 [\phi^+_j,\phi^-_k]_+=0,\quad [\phi^+_j,\phi^+_k]_+ = [\phi^-_j,\phi^-_k]_+ =(-1)^j \delta_{j+k,0}.
 \label{neutral-canonical}
\ee
In particular,
\be
(\phi^+_0)^2=(\phi^-_0)^2=\tfrac{1}{2}.
\label{square_phi_0}
\ee
Acting on the vacua $|0\rangle$ and $|1 \rangle$, we have
\bea
\phi^+_{-j}|0\rangle &\&= \phi^-_{-j} |0\rangle =\phi^+_{-j}|1\rangle = \phi^-_{-j} |1\rangle =0 , \quad  \forall  j > 0  , \quad \forall  j > 0, 
\label{phi_vac_r} \\
\langle 0| \phi^+_{j} &\&= \langle 0|\phi^-_{j} = \langle 1| \phi^+_{j} = \langle 1|\phi^-_{j}  =0 , \quad  \forall  j > 0 ,
\label{phi_vac1_l}
\\
\phi^+_0|0\rangle &\& = 
- i \phi^-_0 |0\rangle =
\tfrac{1}{\sqrt{2}} \psi_0|0\rangle = \tfrac{1}{\sqrt{2}} |1\rangle  ,
\label{phi_0_ac_r}    
\\
 \langle 0| \phi^+_0 &\& = 
 i \langle 0|\phi^-_0  =
\tfrac{1}{\sqrt{2}} \langle  0 | \psi_0^\dag = \tfrac{1}{\sqrt{2}}\langle 1|.
\label{hatphi_0_vac_l}  
\eea
The pairwise expectation values are:
\bea\label{phi_pairing}
    \langle 0| \phi^+_j \phi^+_k |0\rangle&\& =\langle 0| \phi^-_j\phi^-_k|0\rangle =
    \begin{cases}
      (-1)^k\delta_{j,-k}& \text{if}\ k>0,\\
      \tfrac12\delta_{j,0}& \text{if}\ k=0,\\
      0& \text{if}\ k<0,
    \end{cases} \\
\langle 0|\phi^+_j\phi^-_k|0\rangle &\&= -\langle 0|\phi^-_j \phi^+_k|0\rangle =\tfrac {i}{2} \delta_{j,0}\delta_{k,0}
\label{phi_pairing'}.
\eea

%%%%%%%%%%%%%%%%%  Subsection 2. 2 Fermionic Wick theorem  %%%%%%%%%%%%%%%%%%%%%%%%%%
\subsection{Fermionic Wick theorem}
\label{fermionic_wick}

For an even number of fermionic operators $(w_1, \dots, w_{2L})$ that anticommute:
\be
[w_j, w_k]_+=0, \quad 1 \le j, k \le 2L,
\label{anticommute}
\ee
the matrix with elements  $\langle 0|w_jw_k |0\rangle$ is skew symmetric, and Wick's theorem implies that
the vacuum state expectation value  $\langle 0|w_1\cdots w_{2L} |0\rangle$ of the product is given by the Pfaffian
\be
\langle 0|w_1\cdots w_{2L} |0\rangle = \Pf\left( \langle 0|w_jw_k |0\rangle \right)_{1\le j,k \le 2L}.
\label{Wick-Pf}
\ee

On the other hand,  if  the odd elements  $w_1,w_3,\dots$ are linear combinations of creation 
operators $\{\psi_j\}_{j\in\Zb}$ and the even ones
$w_2, w_4,\dots$ linear combinations of annihilation operator $\{\psi_j^\dag\}_{j\in\Zb}$,  Wick's theorem implies
\be
\langle 0|w_1\cdots w_{2L} |0\rangle = \det\left( \langle 0|w_jw_k |0\rangle \right)_{j=1,3,\dots; k=2,4,\dots}.
\label{Wick-det}
\ee

%%%%%%%%%%%%%  Subsection 2. 3 Current components and a factorization theroem %%%%%%%%%%%%%%%%%%
\subsection{Current components and a factorization lemma}
\label{currents_factorization}

The positive current components of charged fermions, defined as
\be
 J_n =\sum_{i\in\mathbb{Z}} \psi_i \psi^\dag_{i+n},\quad n\in \Nb^+,
 \label{J_n_def}
\ee
commute amongst themselves
\be
[J_n, J_m ] =0,\quad n,m \in \Nb^+,
\ee
and generate the KP flows \cite{Sa, JM}.

The neutral fermion current components $J^B_n$ and $\hat{J}^B_n$ are defined as
\be
J^{B+}_n:=\tfrac 12 \sum_{j\in\mathbb{Z}} (-1)^{j+1}\phi^+_j \phi^+_{-j-n}\,,\quad 
J^{B-}_n:=\tfrac 12 \sum_{j\in\mathbb{Z}} (-1)^{j+1}\phi^-_j\phi^-_{-j-n}, \quad n\in \Nb^+.
\label{J_pq_B_def}
\ee
The even components $\{J^{B+}_{2p}, J^{B-}_{2p}\}$ all vanish, 
while the odd ones mutually commute:
\be
[J^{B+}_{2p-1}, J^{B+}_{2q-1}]  = 0,\quad [J^{B-}_{2p-1}, J^{B-}_{2q-1}]=0,\quad
[J^{B+}_{2p-1}, J^{B-}_{2q-1}] =0, \quad p,q \in \Nb^+ ,
\ee
and both generate BKP flows \cite{DJKM1, DJKM2, JM}.

By (\ref{vac_annihil_psi_j_r}), the positive  current components $J_n$  annihilate the vacuum $|0\rangle$:
\be
J_n |0 \rangle =0, \quad n\in \Nb^+
\label{J_n_vac_0}
\ee
and,  by (\ref{phi_vac_r}), the neutral current components $J^B_n$ and $\hat{J}^B_n$  annihilate the vacua $|0\rangle$ and $|1\rangle$:
\be
J^{B+}_n|0 \rangle = 0, \quad J^{B-}_n|0 \rangle =0, \quad J^{B+}_n|1 \rangle = 0, \quad J^{B-}_n|1 \rangle, \quad n\in \Nb^+.
\label{J_B_vac}
\ee
It  also follows from (\ref{J_n_def}),  (\ref{J_pq_B_def}) and 
\be
\psi_j =\frac{\phi^+_j-i\phi^-_j}{\sqrt 2},\quad 
\psi^\dag_{-j} =(-1)^j\frac{\phi^+_j+i\phi^-_j}{\sqrt 2},
\label{psi_j_phI_j}
\ee
that:
\bl
For  odd $n= 2p-1$, 
\be
 J_{2p-1}=J^{B+}_{2p-1}+J^{B-}_{2p-1}, \quad p \in \Nb^+.
 \label{J_gamma_gamma_n}
\ee
\el

The following factorization property of VEV's is proved in \cite{JM} and \cite{You}.
\begin{lemma}
If  ${\bf a}^+$ and ${\bf a}^-$ are (finite or infinite) sums of monomials in the $\{\phi^+_i\}_{i\neq 0}$ of even  and odd degrees, respectively,
and  $\hat{{\bf a}}^+$ and $\hat{{\bf a}}^-$ are sums of monomials in the $\{\phi^-_i\}_{i\neq 0}$ of even  and odd degrees, respectively,
then
\bea
\langle 0 | ({\bf a}^+ + \phi^+_0  {\bf a}^-)(\hat{{\bf a}}^+ +\phi^-_0  \hat{{\bf a}}^-) |0 \rangle
&\&= \langle 1 | ({\bf a}^+ + \phi^+_0  {\bf a}^-) (\hat{{\bf a}}^+ +\phi^-_0  \hat{{\bf a}}^-)|1 \rangle  \cr
&\& = \langle 0 | {\bf a}^+  |0 \rangle\langle 0 | \hat{{\bf a}}^+  |0 \rangle.
\eea
\end{lemma}
As an immediate corollary, it follows that:

  \begin{lemma}[\bf Factorization]
 \label{factorization_lemma}
If $(u^-_1,\dots,u^+_n)$ and $(u^-_1,\dots,u^-_m)$ are linear combinations of the
operators $\{\phi^+_i\}_{i \in \Zb}$ and  $\{\phi^-_i\}_{i \in \Zb}$ respectively, the VEV of their product can be factorized as:
\be
\langle 0 | u^+_1\cdots u^+_n u^-_1\cdots u^-_m |0\rangle =
\begin{cases}
 \langle 0 | u^+_1\cdots u^+_n  |0\rangle \langle 0|u^-_1\cdots u^-_m |0\rangle 
&  
\text{ if } n \text{ and } m \text{ are even }\\
0  &  \text{ if } n \text{ and }  m\text{ have different  parity }\\ 
2i \langle 0|u^+_1\cdots u^+_n\phi^+_0 |0\rangle \langle 0|u^-_1\cdots u^-_m\phi^-_0|0\rangle &  
\text{  if }  n \text{ and }  m \text{ are odd }.
\end{cases}
\label{factorization}
\ee
In particular, for odd $n$ and $m$,  $\langle 0 |u^+_1\cdots u^+_n \phi^+_0|0\rangle$
and $\langle 0|u^-_1\cdots u^-_m \phi^-_0 |0\rangle $  vanish unless the terms in the products
$u^+_1\cdots u^+_n$ and $u^-_1\cdots u^-_m$  that are linear in $\phi^+_0$ and $\phi^-_0$, respectively, are nonzero.
\end{lemma} 
%%%%%%%%%%%%%%%%%  Section 3. Schur and  Schur $Q$-functions via fermions.%%%%%%%%%%%%%%%%%%%%%%%

\section{Schur functions and Schur $Q$-functions via fermions.}
\label{schur_schur_Q_fermions}

%%%%%%%%%%%%%%%%%  Section 3.1 Fermioinc representation of Schur functions %%%%%%%%%%%%%%%%%%%%%%%
\subsection{Fermionic representation of KP flows and Schur functions.}
\label{schur_fermions}

Let $\tb=(t_1,t_2,t_3,\cdots )$ denote the infinite sequence of KP flow parameters, and define the abelian
group $\Gamma_+$ of KP flows  $\{\gamma_+(\tb) := e^{\sum_{i=1}^\infty t_i \Lambda^i}\}$,  where
$\Lambda \in \End(\HH)$ is the upward shift element
\be
\Lambda(e_i) = e_{i-1}.
\ee
These act on $\FF$ via the fermionic representation
\be
\hat{\gamma}_+(\tb) := e^{\sum_{n=1}^\infty J_n t_n }.
\ee
and, by (\ref{J_n_vac_0}),  they stabilize the vacuum element
 \be
 \hat{\gamma}_+(\tb)| 0 \rangle = |0 \rangle.
 \ee
 
  We  define  $\psi_j(\tb)$, $\psi^\dag_j(\tb)$ to be the conjugates of $\psi_j$  and $\psi^\dag_j$
  by $\hat{\gamma}_+(\tb) $.
\be
 \psi_j(\tb):=\hat{\gamma}_+(\tb) \psi_j (\hat{\gamma}_+(\tb))^{-1},\quad
 \psi^\dag_j(\tb):=\hat{\gamma}_+(\tb) \psi^\dag_j (\hat{\gamma}_+(\tb))^{-1}.
\ee
It follows that these satisfy the same anticommutation relations (\ref{charged-canonical}) as $\{\psi_j, \psi^\dag_j\}_{j\in \Zb} $:
\be
[\psi_j(\tb),\psi_k(\tb)]_+= [\psi^\dag_j(\tb),\psi^\dag_k(\tb)]_+=0,\quad [\psi(\tb)_j,\psi^\dag_k(\tb)]_+=\delta_{jk}.
\label{charged-canonical-dressed}
\ee

Let  $ (\alpha|\beta)$ be the Frobenius notation for a partition $\lambda$,
with  $\alpha=(\alpha_1,\dots,\alpha_r)$ and $\beta=(\beta_1,\dots,\beta_r)$ the ``arm'' and ``leg''
lengths, respectively, along the principal diagonal of the corresponding Young diagram, and $r$ 
 the Frobenius {\em rank} (i.e., the number of elements on the principal diagonal).
Viewed as functions of the normalized power sum symmetric functions \cite{Mac} in some auxiliary set of
variables $\{x_a\}_{a\in \Nb}$ 
\be
t_i :=\tfrac{1}{i}  p_i = \tfrac{1}{i}  \sum_{a=1}^\infty x_a^i,
\label{t_i_sum_x_a}
\ee
we may express the corresponding Schur function $s_{(\alpha|\beta)}(\tb)$ \cite{Mac},
viewed as functions of the normalized power sums $\tb$, defined in (\ref{t_i_sum_x_a}),
as the following vacuum state expectation values  \cite{Sa, JM}
\bea
s_{(\alpha|\beta)}(\tb) &\& =  (-1)^{\sum_{j=1}^r\beta_j}(-1)^{\tfrac{1}{2}r(r-1)}
 \langle 0|\psi_{\alpha_1}(\tb)\cdots \psi_{\alpha_r}(\tb)
 \psi_{-\beta_1-1}^\dag(\tb)\cdots \psi_{-\beta_r-1}^\dag(\tb)  |0\rangle \cr
&\&= (-1)^{\sum_{j=1}^r\beta_j} \langle 0|\prod_{j=1}^r 
 \left( \psi_{\alpha_j}(\tb) \psi_{-\beta_j-1}^\dag(\tb) \right) |0\rangle ,
  \label{schur_fermionic}
\eea
where,  by (\ref{charged-canonical-dressed}),  
\be
[\psi_{\alpha_k}(\tb), \psi_{-\beta_j-1}^\dag(\tb)]_+= 0, \quad   \forall j,k\in  \Zb.  
\ee

 \br{\bf Giambelli identity.}
 \label{Giambelli} 
Applying Wick's theorem (\ref{Wick-det}) to the right hand side of (\ref{schur_fermionic})
 gives the Giambelli identity \cite{Mac} 
 \bea
s_{(\alpha|\beta)}({\bf t}) &\&= 
\det\left((-1)^{\beta_j} \langle 0| \psi_{\alpha_k}(\tb)\psi_{-\beta_j-1}^\dag(\tb) |0\rangle  \right)_{1\le j, k, \le r}\cr
&\& \cr
&\& = \det\left(s_{(\alpha_j|\beta_k)}({\bf t}) \right)_{1\le j,k \le r}, 
\eea
expressing  $s_{(\alpha|\beta)}({\bf t})$ as the determinant of the $r \times r$
matrix formed from the hook partition Schur functions for each pair of Frobenius indices.
\er

%%%%%%%%%%  Section 3.2 Fermionic representations of BKP flows and Schur $Q$-functions %%%%%%%%%%%%%%%%%
\subsection{Fermionic representation of BKP flows and Schur $Q$-functions}
\label{schur_Q_fermions}

Denoting  the set of odd flow variables as 
\be
\tb_B=(t_1,t_3,t_5,\dots),
\ee
these may be viewed as determining a subset  $\{\tb'\} $ of the $\tb$'s , where
\be
\tb':=(t_1,0,t_3,0,t_5,\dots).
\ee
Following \cite{DJKM1, DJKM2,  You}, we  define  two mutually commuting abelian groups of BKP flows 
\hbox{$\Gamma^{B+} = \{\gamma^{B+}({\tb}_{B+})\}$} and 
\hbox{$\Gamma^{B-} = \{\gamma^{B-}({\tb}_B)\}$}, 
with Clifford representations
\be
\hat{\gamma}^{B+}({\tb}_B) :=  e^{\sum_{p=0,}^\infty J^{B+}_{2p-1} t_{2p-1} } ,
 \quad \hat{\gamma}^{B-}({\tb}_B) := e^{\sum_{q=1}^\infty J^{B-}_{2q-1} t_{2q-1} }, 
\ee
and note that, by (\ref{J_B_vac}), these  stabilize both the vacua $|0\rangle$ and $| 1\rangle$
\bea
\hat{\gamma}^{B+}({\tb}_B)|0\rangle &\&= |0\rangle, \quad \gamma^{B+}({\tb}_B)|1\rangle = |1\rangle, \cr
\hat{\gamma}^{B-}({\tb}_B)|0\rangle &\&= |0\rangle, \quad \hat{\gamma}^{B-}({\tb}_B)|1\rangle = |1\rangle.
\eea
Defining
\be
 \phi^+_j(\tb_B):=\hat{\gamma}^{B+}(\tb_B) \phi^+_j (\hat{\gamma}^{B+}(\tb_B))^{-1}, \quad
\phi^-_j(\tb_B):=\hat{\gamma}^{B-}({\tb}_B) \phi^-_j ((\hat{\gamma}^{B-}({\tb}_B) )^{-1},
\ee
it follows that these satisfy the same anticommutation relations (\ref{neutral-canonical}) as $\{\phi^+_j, \phi^-_k\}_{j,k \in \Zb}$:
\be
 [\phi^+_j(\tb_B),\phi^-_k(\tb_B)]_+=0,\quad [\phi^+_j(\tb_B),\phi^+_k(\tb_B)]_+ =
 [\phi^-_j(\tb_B),\phi^-_k(\tb_B)]_+ =(-1)^j \delta_{j+k,0}.
 \label{neutral-canonical-dressed}
\ee
From eq.~(\ref{J_gamma_gamma_n}), we have
\begin{lemma}
\be
\psi_j(\tb') = \frac{\phi^+_j(\tb_B)-i\phi^-_j(\tb_B)}{\sqrt 2}, 
\quad
\psi^\dag_{-j}(\tb') =  (-1)^j\frac{\phi^+_j(\tb_B)+i\phi^-_j(\tb_B)}{\sqrt 2} .
\label{f-dag-b-dressed'}
\ee
\end{lemma}

\br
\label{undressed-vs-dressed}
It follows from eq.~(\ref{J_B_vac}) that (\ref{phi_vac_r}) - (\ref{phi_0_ac_r}) are still valid if $\phi^+_0$ and $\phi^-_0$
are replaced by $\phi^+_0({\tb }_B)$ and $\phi^-_0({\tb }_B)$, so we have 
\be
\phi^+_0({\tb }_B)|0\rangle = \phi^+_0|0\rangle = \tfrac{1}{2} |1\rangle, \quad \phi^-_0({\tb }_B)|0\rangle 
= \phi^-_0| 0\rangle = \tfrac{i}{2} |1\rangle.
\label{phi_0_t_vac}
\ee
\er

From (\ref{f-dag-b-dressed'}), it follows that (\ref{schur_fermionic}) may equivalently be expressed as:
\begin{lemma}
\label{key-Lemma}
\bea
s_{(\alpha| \beta)} ({\bf t}')&\&=(-1)^{\tfrac{1}{2}r(r+1)}
2^{-r}  \langle 0|  
  \left(\phi^+_{\alpha_1}(\tb_B)-i\phi^-_{\alpha_1}(\tb_B)\right)\cdots 
\left(\phi^+_{\alpha_r}(\tb_B)-i\phi^-_{\alpha_r}(\tb_B)\right)
 \cr
&\& 
{\hskip 12 pt}\times 
\left(\phi^+_{\beta_1+1}(\tb_B)+i\phi^-_{\beta_1+1}(\tb_B)\right)\cdots 
 \left(\phi^+_{\beta_r+1}(\tb_B)+i\phi^-_{\beta_r+1}(\tb_B)\right)  |0\rangle. 
\label{s_alpha_beta_neutral_fermion}
\eea
\end{lemma}

To define Schur's $Q$-functions $Q_\alpha$ we  begin, following \cite{Mac},  by defining an infinite skew 
symmetric matrix $(Q_{ij})_{i, j \in \Nb}$, whose entries are symmetric functions of the infinite sequence 
of indeterminates ${\bf x} =(x_1, x_2, \dots )$, via the following formula:
\be
Q_{ij}({\bf x}) := 
\begin{cases}
q_i({\bf x}) q_j({\bf x}) + 2\sum_{k=1}^j (-1)^kq_{i+k}({\bf x}) q_{j-k}({\bf x}) \quad \text{if } (i,j) \neq (0,0), \\
0  \quad \text{if } (i,j) = (0,0) ,
\end{cases} 
\ee
where the $q_i({\bf x})$'s are defined by the generating function:
\be
\prod_{i=1}^\infty {1+ z x_i \over 1 - z x_i}= \sum_{i=0}^\infty z^i q_i({\bf x}).
\ee
In particular
\be
Q_{(j,0)}\left({\bf x}\right) = -Q_{(0,j)}\left({\bf x}\right)=q_j\left({\bf x}\right) \ \text{ for } j\ge 1.
\ee

For a strict partition $\alpha$ of even cardinality $r$ (including a possible zero part $\alpha_r=0)$, 
let $\Mb_{\alpha}({\bf x})$ denote the $r\times r$ skew symmetric  matrix with entries
\be
\left(\Mb_{\alpha}({\bf x})\right)_{ij} := Q_{\alpha_i \alpha_j}({\bf x}), \quad 1\le i, j \le r.
\ee    
The Schur $Q$-function is defined as its Pfaffian \cite{Mac}
\be
Q_\alpha({\bf x}):= \Pf(\Mb_{\alpha}({\bf x}))
\label{Q+_pfaff}
\ee
and, for completeness, 
\be
Q_{\emptyset} :=1.
\ee

Equivalently,  these may be viewed as functions  of the  odd (normalized) power sum symmetric functions 
${\bf t}_B=(t_1, t_3, \dots)$
\be
t_{2i-1} := {1\over 2i-1}p_{2i-1}({\bf x}) = {1\over 2i-1} \sum_{a=1}^\infty x_a^{2i-1},  \quad i =1,2 , \dots.
\ee
which we denote
\be
 \tilde{q}_j({\bf t}_B):= q_j(\xb), \quad  \tilde{Q}_{ij}({\bf t}_B) :=Q_{ij}(\xb).
\ee
We then have the fermionic VEV formulae  \cite{DJKM1, You, NimO}:
\bea
 \tilde{Q}_\alpha\left(\tfrac{1}{2}\tb_B\right) &\& =2^{r\over 2}\langle 0| \phi^+_{\alpha_1}(\tb_B)\cdots \phi^+_{\alpha_{r}}(\tb_B) |0\rangle,
 \label{Q+fermi}
\\
 &\&= 2^{r\over 2}
 \langle 0| \phi^-_{\alpha_1}(\tb_B)\cdots \phi^-_{\alpha_{r}}(\tb_B) |0\rangle,
 \label{Q-fermi}
\eea
which follow from  the Pfaffian form (\ref{Wick-Pf}) of Wick's theorem.

%%%%%%%%%%%%%%%%%%  Subsection 3.3. Examples: ``doubles'', hook partitions and an $r=2$ case  %%%%%%%%
\subsection{Examples:  ``doubles'', hook partitions and an $r=2$ case}
\label{example_doubles}

Consider a set of Frobenius indices
\be
\alpha=(\alpha_1,\dots,\alpha_{r}),
\ee
with $\alpha_i > \alpha_{i+1}, \ \alpha_r \ge 0$  (or strict partitions, with $\alpha_r=0$ allowed as a part).
Following \cite{Mac}, let $\DP$ denote the set of strict partitions (with all parts $\ge 1$).
Associated to $\alpha$ we define the strict partition
\be
I(\alpha)= (I_1(\alpha), \dots, I_r(\alpha)) \in \DP
\ee
whose parts are obtained form the $\alpha_i$'s by shifting upward by $1$:
\be
I_i(\alpha) := \alpha_i + 1.
\label{I_alpha}
\ee
If a partition $\lambda$ is related to a strict partition $I =(I_1, \dots I_r)$ in such a way that
its Frobenius indices are
 \be
 \lambda = (I_1, \dots I_r |I_1 -1, \dots I_r-1)) ,
 \ee
 it is called \cite{Mac} the {\em double of} $I$, and denoted $\lambda =D(I)$.

 \begin{example}
 \label{example3.1}
It is known (see \cite{Mac}, Chapter  III, Section 8, example 10 (b)) that
 \be
s_{D(\alpha)}(\tb') =
\begin{cases}
 2^{-r}\left(\tilde{Q}_{\alpha}\left(\tfrac12 \tb_B\right)  \right)^2 \ \text{ if } r \ \text{is even},\cr
 2^{-r}\left(\tilde{Q}_{(\alpha, 0)} \left(\tfrac12 \tb_B\right)  \right)^2 \ \text{ if } r \ \text{is odd}, 
 \end{cases} 
 \label{schur_D_alpha_Q}
 \ee
 where $D(\alpha)$  is the double of $\alpha$.
To prove this using fermionic VEV's, we use:
\bea
&\&(-1)^{\beta_j+1}  \psi_{\alpha_j}(\tb') \psi_{-\beta_j-1}^\dag(\tb')  = -\frac12
\left(\phi^+_{\alpha_j}(\tb_B)-i\phi^-_{\alpha_j}(\tb_B)\right) 
\left( \phi^+_{\beta_j+1}(\tb_B)+i\phi^-_{\beta_j+1}(\tb_B)\right) 
\cr
  =  &\&\frac12 \left(\phi^+_{\beta_j+1}(\tb_B)\phi^+_{\alpha_j}(\tb_B) +  \phi^-_{\beta_j+1}(\tb_B)\phi^-_{\alpha_j}(\tb_B)\right)
 +\frac i2 \left(
     \phi^-_{\beta_j+1}(\tb_B)\phi^+_{\alpha_j}(\tb_B)
    -  \phi^+_{\beta_j+1}(\tb_B)\phi^-_{\alpha_j}(\tb_B) \right),
    \cr
    &\&
    \label{psidag_psi_phihat_phi}
\eea
so for $\alpha_j=\beta_j+1$, we have 
\be
(-1)^{\alpha_j} \psi_{\alpha_j}(\tb')  \psi_{-\alpha_j}^\dag(\tb')  
=  - i  \phi^+_{\alpha_j}(\tb_B)  \phi^-_{\alpha_j}(\tb_B).
 \label{psidag_psi_alpha_phihat_phi}
\ee
From eqs.~(\ref{schur_fermionic}) and (\ref{psidag_psi_alpha_phihat_phi}) it follows  that
\bea
s_{D(\alpha)}(\tb') &\&=  (-i)^{r} \langle 0|  \left( \prod_{m=1}^r  
\phi^+_{\alpha_j}(\tb_B) 
\phi^-_{\alpha_j}(\tb_B) \right)
   |0\rangle \cr
   &\& =
   (i)^r (-1)^{\tfrac{1}{2}r(r+1)}\langle 0| \left(\phi^+_{\alpha_1}(\tb_B)\cdots \phi^+_{\alpha_r}(\tb_B) \phi^-_{\alpha_1}(\tb_B)\cdots \phi^-_{\alpha_r}(\tb_B) \right)  |0\rangle .
   \label{schur_double_fermi}
   \eea
Applying  Lemma \ref{factorization_lemma}, eqs.~(\ref{Q+fermi}) and (\ref{Q-fermi})
gives (\ref{schur_D_alpha_Q}).
\end{example}

\begin{example}
\label{example3.2}
It is also known (see \cite{Mac}, Chapter  III, Section 8, example 10 (a, iv) ) that
\bea
 s_{(j | k)} (\tb') &\&=  \tfrac{1}{2}\left( \tilde{Q}_{j}(\tfrac12\tb_B)  \tilde{Q}_{k+1}(\tfrac12\tb_B) -  \tilde{Q}_{j,k+1}(\tfrac12\tb_B)\right) \cr
 &\& \cr
  &\&= \tfrac{1}{2}\left( \tilde{Q}_{(j,0)}(\tfrac12\tb_B)  \tilde{Q}_{(k+1, 0)}(\tfrac12\tb_B) -  \tilde{Q}_{j,k+1}(\tfrac12\tb_B)\tilde{Q}_{\emptyset}(\tfrac12\tb_B)\right).
 \label{schur_i_j_Q}
 \eea
To derive this fermionically, we apply (\ref{schur_fermionic}), which gives:
\bea
s_{(j|k)}({\bf t}') &\&= 
- \tfrac{1}{2}\langle 0 | (\phi^+_j({\bf t}_B) - i \phi^-_j({\bf t}_B))
 (\phi^+_{k+1}({\bf t}_B) + i \phi^-_{k+1}({\bf t}_B)  | 0\rangle  \cr
&\& = \tfrac{1}{2}\langle 0 | \phi^+_{k+1}({\bf t}_B) \phi^+_j ({\bf t}_B)| 0 \rangle 
+ \tfrac{1}{2}\langle 0 | \phi^-_{k+1}({\bf t}_B) \phi^-_j({\bf t}_B) | 0 \rangle \cr
&\& \quad+  \tfrac{i}{2}\langle 0 | \phi^-_{k+1}({\bf t}_B) \phi^+_j | 0 \rangle - \tfrac{i}{2}\langle 0 | \phi^+_{k+1} \phi^-_j({\bf t}_B) | 0 \rangle \cr&\& = \tfrac{1}{2}\langle 0 | \phi^+_{k+1}({\bf t}_B)\phi^+_j ({\bf t}_B)| 0 \rangle
 + \tfrac{1}{2}\langle 0 | \phi^-_{k+1}({\bf t}_B)\phi^-_j({\bf t}_B) | 0 \rangle \cr
&\& \quad+ \langle 0 | \phi^-_{k+1}({\bf t}_B)\phi^-_0({\bf t}_B)|0 \rangle \langle  0 |\phi^+_j ({\bf t}_B)\phi^+_0({\bf t}_B)| 0 \rangle
 +\langle 0 | \phi^+_{k+1}({\bf t}_B) \phi^+_0({\bf t}_B)|0 \rangle \langle 0 | \phi^-_j({\bf t}_B)\phi^-_0({\bf t}_B)| 0 \rangle \cr
&\& =  \tfrac{1}{2}\left( \tilde{Q}_{(j,0)}(\tfrac12\tb_B)  \tilde{Q}_{(k+1, 0)}(\tfrac12\tb_B) -  \tilde{Q}_{j,k+1}(\tfrac12\tb_B) \tilde{Q}_{\emptyset}(\tfrac12\tb_B)\right),
\eea
where  Lemma \ref{factorization_lemma} has been used in the third equality,  and eqs.~(\ref{phi_0_t_vac}),
(\ref{Q+fermi}) and (\ref{Q-fermi}) in the last.
\end{example}

\begin{example}
\label{example3.3}
Consider a partition  $\lambda =(\alpha_1, \alpha_2 | \beta_1, \alpha_2-1 )$  of Frobenius rank $r=2$
 in which $\alpha_1> \beta_1+1 >\alpha_2$ and   $\beta_2 = \alpha_2-1$. Expanding the RHS of eq.~(\ref{s_alpha_beta_neutral_fermion})
 for this case,  collecting together the four types of terms:
 \bea
&\& \langle 0|\phi^+_{\alpha_1}({\bf t}_B)\phi^+_{\alpha_2}({\bf t}_B)\phi^-_{\beta_1+1}({\bf t}_B)\phi^-_{\alpha_2}({\bf t}_B) |0\rangle, \quad
\langle 0| \phi^+_{\alpha_1}({\bf t}_B)\phi^-_{\alpha_2}({\bf t}_B)\phi^-_{\beta_1+1}({\bf t}_B)\phi^-_{\alpha_2}({\bf t}_B) |0\rangle, \cr
 &\&\langle 0|\phi^+_{\alpha_1}({\bf t}_B)\phi^+_{\alpha_2}({\bf t}_B)\phi^-_{\beta_1+1}({\bf t}_B)\phi^+_{\alpha_2}({\bf t}_B) |0\rangle,  \quad
\langle 0|  \phi^+_{\alpha_1}({\bf t}_B)\phi^-_{\alpha_2}({\bf t}_B)\phi^-_{\beta_1+1}({\bf t}_B)\phi^+_{\alpha_2}({\bf t}_B) |0\rangle, \cr
&\&
\eea
 applying  Lemma \ref{factorization_lemma} and using eqs.~(\ref{Q+fermi}) and (\ref{Q-fermi}) gives
\be
 s_{(\alpha_1,\alpha_2|\beta_1,\alpha_2-1)}({\bf t}')  
 =\frac14  \tilde{Q}_{(\alpha_1,\alpha_2)}(\tfrac12{\bf t}_B) \tilde{Q}_{(\beta_1+1,\alpha_2)}(\tfrac12{\bf t}_B)
  -\frac14 \tilde{Q}_{(\alpha_1,\beta_1+1,\alpha_2, 0)}(\tfrac12{\bf t}_B) \tilde{Q}_{(\alpha_2, 0)}(\tfrac12{\bf t}_B).
\label{s_r2_s1}
\ee
\end{example}

Following some preparatory definitions in Section \ref{polariz_binary},  Theorem \ref{S_QQ},  Section \ref{schur_Q_schur_bilin}   
provides a generalization of these  identities expressing $s_{(\alpha , \beta)}(\tb')$, 
for arbitrary partitions $(\alpha | \beta)$,  as sums over products of Schur $Q$-functions.

%%%%%%%%%%%  Subsection 4 Intersections, unions, polarizations and binary markings %%%%%%%%%%%%%%%%%%
\section{Polarizations and binary markings}
\label{polariz_binary}

Let $(\alpha | \beta)$ be a partition of Frobenius rank $r$.
Denote the  union and intersection of $\alpha$ with $I(\beta)$ as
\be
S := \alpha \cap I(\beta), \quad T := \alpha \cup I(\beta),
\label{S_T_def}
\ee
and their cardinalities  (or lengths, when viewed as strict partitions) as
\be
s :=\#(S) = \ell(S), \quad t:= \#(T) = \ell(T) = 2r-s.
\ee

\begin{definition}{{\bf Polarizations.}}
 \label{polarization}
For any partition $(\alpha|\beta)$ a {\em polarization} is a pair of strict partitions 
\be
\mu :=(\mu^+, \mu^-)
\ee  
(with $0$ 's allowed as parts), such that the following conditions are satisfied
\be
\mu^+\cap \mu^- = S  =\alpha \cap I(\beta), \quad \mu^+\cup \mu^- = T =\alpha \cup I(\beta).
\label{polariz_S_T}
\ee
\end{definition}
Let  $\PP(\alpha,\beta)$ denote the set of all polarizations corresponding to a partition $(\alpha|\beta)$ . 
The following Lemma shows that the cardinality of $\PP(\alpha,\beta)$ is $2^{2r-2s}$.
\begin{lemma}
\label{polariz_count}
The number of  distinct polarizations $(\mu^+, \mu^-)$ corresponding to a pair of strict partitions $S \ss T$, 
with $T$ of cardinality $2r-s$ and $S$ of cardinality $s$  is $2^{2r -2s}$.
\end{lemma}
\begin{proof}
The total number of elements in $T$  is $2r-s$, and the number of these that are in $S$,
and hence in both $\mu^+$ and $\mu^-$ is $s$.
The remaining $2(r-s)$ elements are either in one or the other of the two strict partitions 
$\mu^{\pm}$, but not both,  and hence there are  $2^{2(r -s)}$ distinct ways to
select the polarization $(\mu^+, \mu^-)$. 
\end{proof}

Let
\be
m^+(\mu) := \#(\mu^+) , \quad m^-(\mu) =\#(\mu^-) ,
\ee
denote the cardinalities of $\mu^+$ and $\mu^-$.
Their sum is 
\be
m^+(\mu) + m^-(\mu) =  2r,
\label{m+m-}
\ee
and therefore, they are either both even or both odd.
Denoting the cardinalities of the  intersections $\alpha \cap \mu^-$ and  $I(\beta) \cap \mu^- $
\be
\pi(\mu) :=\#(\alpha \cap \mu^-), \quad \tilde{\pi}(\mu) := \#(I(\beta) \cap \mu^- ),
\label{pi_j_def}
\ee
it follows that 
\be
\pi(\mu)+\tilde{\pi}(\mu) = m^-(\mu) +  s.
\label{pi_tildepi}
\ee

\begin{definition}{\bf Binary markings.}
For any integer $j$ between $0$ and $2^{2r}-1$, we have the associated binary sequence
\be 
\epsilon(j) :=(\epsilon_1(j), \dots, \epsilon_{2r}(j)),  
\ee
where $\epsilon_k(j) =+$ if the $k$th element in the binary representation of $j$ is $0$
and $\epsilon_k(j) =-$ if the $k$th element is $1$.
We call the product
\be
\phi^{(j)}(\alpha, \beta)
:=\phi^{\epsilon_1(j)}_{\alpha_1} \cdots \phi^{\epsilon_r(j)}_{\alpha_r}  \phi^{\epsilon_{r+1}(j)}_{\beta_1+1} 
\cdots \phi^{\epsilon_{2r}(j)}_{\beta_r+1} 
\label{binary_marking}
\ee
the $j$th {\em binary marking} of the sequence $(\alpha, I(\beta))$.
\end{definition}
If any two factors in (\ref{binary_marking}) coincide, the product vanishes. The number of these is $2^{2r-s}(2^s-1)$,
and the number of nonvanishing $\phi^{(j)}(\alpha, \beta)$'s is $2^{2r-s}$.
For the latter, there are $s$ pairs $(m, n)$ of type
$(\phi^+_{\alpha_m}, \phi^-_{\beta_n +1})$ and $(\phi^-_{\alpha_m}, \phi^+_{\beta_n +1})$ 
where $\alpha_m=\beta_n+1 \in S$.
 By  reordering the product (\ref{binary_marking}) so that all the $\phi^+$'s appear to the left, 
  and the $\phi^-$'s appear to the right,  with all the subscripts in each group in decreasing order,
every binary sequence $\epsilon(j)$  for which $\phi^{(j)}(\alpha, \beta)\neq 0$ 
determines a unique polarization $\mu(j) = (\mu^+(j), \mu^-(j))$ such that
\be
\phi^{(j)}(\alpha, \beta) =: 
\pm \phi^+_{\mu^+_1}  \cdots \phi^+_{\mu^+_{m^+(\mu)}} \phi^-_{\mu^-_1} \cdots \phi^-_{\mu^-_{m^-(\mu)}}.
  \label{perm_sign_mu_i}
\ee
Of these, the polarization determined by 
\be
\epsilon(2^r-1) = (\underbrace{+, \cdots +}_{r \ \text{terms}}, \underbrace{-, \cdots -}_{r \ \text{terms}}), 
\label{canon_polariz}
\ee
is
\be
(\mu^+(2^r-1), \mu^-(2^r-1)) := (\alpha,I(\beta)),
\ee
and this will be referred to as the {\em canonical polarization}. 

The $2^{2r-s}$ binary sequences  $\epsilon(j)$  for which $\phi^{(j)}(\alpha, \beta)\neq 0$
may be divided into $2^{2r-2s}$ equivalence classes $[\epsilon(j)]$, each containing $2^{s}$ elements, for which the binary markings 
are equal, within a sign, and hence the polarization is the same:
\be
[\epsilon(j)] = [\epsilon(\tilde{j})] \quad \text{if and only if} \quad \mu(j)=\mu(\tilde{j}).
\ee
There is thus a bijection between the set $\PP(\alpha,\beta)$ of polarizations and the set of equivalence classes
$\{[\epsilon(j)]\}$ of binary  sequences which define (within a sign) the same nonvanishing
binary markings $\phi^{(j)}(\alpha, \beta)$.
The $2^s$ elements of each equivalence class $[\epsilon(j)]$ are related by interchanging any number of
 the $s$ pairs $(m, n)$ of type
\be
(\phi^+_{\alpha_m}, \phi^-_{\beta_n +1}) \leftrightarrow (\phi^-_{\alpha_m}, \phi^+_{\beta_n +1}),
\label{phi+_phi-_inter}
\ee 
where $\alpha_m=\beta_n+1 \in S$. 

 The set of all $\alpha_m$'s and $\beta_n+1$'s appearing in (\ref{binary_marking}) with a $+$ superscript, 
written in decreasing order, is the strict partition $\mu^+(j)$  forming the first part of the polarization $\mu(j)$ 
and the set of all $\alpha_m$'s and $\beta_n+1$'s in appearing with a $-$ superscript,
 also written in decreasing order, is the second part $\mu^-(j)$.

In every equivalence class $\epsilon(j)$, there is a unique element $\epsilon(j_0)$, for which 
\be
\epsilon_{\alpha_m}(j_0) = +,  \quad \forall \ \alpha_m \in S.
\ee
which will be referred to as the {\em canonical representative}.
 \begin{remark}
There is one canonical representative $\epsilon(j_0)$  in each equivalence class $[\epsilon(j)]$, 
so the number of these  is equal to the number $2^{2r-2s}$  of (nontrivial) equivalence classes, and they
 are in bijective correspondence with the  polarizations $\mu\in \PP(\alpha, \beta)$.
\end{remark}
 
 Let $\sigma(j)$ denote the parity of the number of $\alpha_m$'s in the binary marking (\ref{binary_marking}) that are in $S$ and
 appear with  a superscript $-$, which equals the number of exchanges (\ref{phi+_phi-_inter}) of elements in the product
 (\ref{binary_marking}) defining $\phi^{(j)}(\alpha, \beta)$ needed to convert it to $\phi^{(j_0)}(\alpha, \beta)$
 \be
 \phi^{(j)}(\alpha, \beta) = \sigma(j) \phi^{(j_0)}(\alpha, \beta),
  \label{sigma_j_def}
 \ee
In particular,  $\sigma(j_0)=1$.
\begin{definition}{\bf Polarization sign.}
\label{sign_perm}
We define the {\em sign of the polarization} $\mu =\mu(j)$, denoted $\sgn(\mu)$
by 
\be
\phi^{(j_0)}(\alpha, \beta) =: 
\sgn(\mu)\ \phi^+_{\mu^+_1}  \cdots \phi^+_{\mu^+_{m^+(\mu)}} \phi^-_{\mu^-_1} \cdots \phi^-_{\mu^-_{m^-(\mu)}}.
  \label{perm_sign_mu_i}
\ee
\end{definition}
It follows that
\be
\phi^{(j)}(\alpha, \beta) = 
\sigma(j)\sgn(\mu)\ \phi^+_{\mu^+_1}  \cdots \phi^+_{\mu^+_{m^+(\mu)}} \phi^-_{\mu^-_1} \cdots \phi^-_{\mu^-_{m^-(\mu)}}.
\ee

\begin{example}
The partition $\lambda = ((2,0),(1,0))$  has 
\bea
\alpha &\&=(2,0), \quad I(\beta) = (2,1) ,\cr
S &\&= (2), \quad T =(2,1,0), \quad r=2, \quad s =1.
\eea
The  elements $j_0$ are $2, 3, 6$ and $7$.
The $2^2=4$ distinct associated polarizations are 
\bea
\mu(2) = \mu(8)&\& = (2,1,0), (2)),  \quad  \mu(3) =\mu(9)= ((2,0),(2,1)) \cr
\mu(6) =\mu(12) &\&=((2,1), (2,0)), \quad  \mu(7)=\mu(13) =((2), (2, 1,0)).
\eea
The values of $\sigma(j)$ and $\sgn(\mu(j))$ for these are:
\bea
\sigma(2) &\&= +, \quad \sigma(8) = -, \quad \sgn(\mu(2))=\sgn(\mu(8)) = +, \cr
\sigma(3) &\&= +, \quad \sigma(9) = -, \quad \sgn(\mu(3))=\sgn(\mu(9)) = +, \cr
\sigma(6) &\&= +, \quad \sigma(12) = -, \quad \sgn(\mu(6))=\sgn(\mu(12)) = - ,\cr
\sigma(7) &\&= +, \quad \sigma(13) = -, \quad \sgn(\mu(7))=\sgn(\mu(13)) = + .
\eea
The vanishing binary markings $\phi^{(j)}(\alpha, \beta)$ correspond to: $j=0, 1,  4, 5, 10,  11, 14, 15$.
\end{example}

\begin{definition}{\bf Supplemented partitions.}
If $\mu$ is a strict partition of cardinality $r$ (with $0$ allowed as a part), 
define the associated {\em supplemented partition} $\hat{\mu}$ to be
\be
\hat{\mu} := \begin{cases}  \mu ,\ \text{ if } r \ \text{ is even}, \cr
                  (\mu,0) , \  \text{ if } r \text{ is odd}.
                  \end{cases}
\ee
\end{definition}
Note that $\hat{\mu}$ is always of even cardinality, but not necessarily
 {\em strict} since, if $r$ is odd, it may possibly have two $0$ parts $(\mu_r=0, 0)$ at the end. 
If $m^\pm(\mu)$ are the cardinalities of $\mu^{\pm}$, we denote by $\hat{m}^\pm(\mu)$ 
the cardinalities of $\hat{\mu}^{\pm}$.

%%%%%%%%%%%  Section 5. Schur functions as sums over products Schur $Q$-functions } %%%%%%%%%%%%%%%%%%
\section{Schur functions as  sums over products Schur $Q$-functions}
\label{schur_Q_schur_bilin}

Our main result expresses  any Schur function, $s_{(\alpha|\beta)}({\bf t}')$, evaluated at $\tb'$, as a sum
 over  products of Schur $Q$-functions.
\begin{theorem} 
\label{S_QQ}
For any partition $(\alpha|\beta)=(\alpha_1,\dots,\alpha_r|\beta_1,\dots,\beta_r)$, we have
 \be
s_{(\alpha|\beta)}(\tb')  = {(-1)^{\tfrac{1}{2}r(r+1)+s} \over 2^{2r-s}} \sum_{\mu\in \PP(\alpha, \beta)} \sgn(\mu)
 (-1)^{\pi(\mu)+\tfrac{1}{2}\hat{m}^-(\mu)}
 \tilde{Q}_{\hat{\mu}^+}(\tfrac{1}{2}\tb_B) \tilde{Q}_{\hat{\mu}^-}(\tfrac{1}{2}\tb_B).
   \label{s_lambda_Q_bilinear_bis}
 \ee
 \end{theorem}
\begin{proof}  
Using Lemma \ref{key-Lemma}, expanding the product in (\ref{s_alpha_beta_neutral_fermion}) 
and using (\ref{sign_perm}) and (\ref{sigma_j_def}) gives
\bea
s_{\alpha|\beta)}({\bf t}') &\&= {(-1)^{\tfrac{1}{2} r(r+1)}\over 2^r}
\sum_{j=0}^{2^{2r}-1}(-1)^{\pi(\mu(j))+s} i^{m^-(\mu(j))} \sigma(j)
\langle 0 | \phi^{(j)}({\tb}_B)(\alpha, \beta) |0 \rangle \cr
&\& = {(-1)^{\tfrac{1}{2} r(r+1)+s}\over 2^{r-s}}
\sum_{\mu\in \PP(\alpha, \beta)}(-1)^{\pi(\mu)} i^{m^-(\mu)} \sgn(\mu) \cr
&\& {\hskip 120 pt} \times \langle 0 | \ \phi^+_{\mu^+_1} ({\tb}_B) \cdots 
\phi^+_{\mu^+_{m^+(\mu)}}({\tb}_B) \phi^-_{\mu^-_1}({\tb}_B) \cdots \phi^-_{\mu^-_{m^-(\mu)}}({\tb}_B)|0 \rangle.
\cr
&\&
\label{exp_phi_alpha_beta}
\eea
If $m^+(\mu)$ and  $m^-(\mu)$ are both even, by Lemma \ref{factorization_lemma} and eq.~(\ref{Q+fermi}) we have
\bea
&\&\langle 0 | \ \phi^+_{\mu^+_1} ({\tb}_B) \cdots 
\phi^+_{\mu^+_{m^+(\mu)}}({\tb}_B) \phi^-_{\mu^-_1}({\tb}_B) \cdots \phi^-_{\mu^-_{m^-(\mu)}}({\tb}_B)|0 \rangle \cr
&\& = \langle 0 | \phi^+_{\mu^+_1}({\tb}_B)  \cdots 
\phi^+_{\mu^+_{m^+(\mu)}}({\tb}_B) |0\rangle
 \langle 0 | \phi^-_{\mu^-_1}({\tb}_B) \cdots \phi^-_{\mu^-_{m^-(\mu)}}({\tb}_B)|0 \rangle \cr
&\& = 2^{-r} \tilde{Q}_{\mu^+} (\tfrac{1}{2}{\bf t}_B) \tilde{Q}_{\mu^-} (\tfrac{1}{2}{\bf t}_B) 
= 2^{-r} \tilde{Q}_{\hat{\mu}^-} (\tfrac{1}{2}{\bf t}_B) \tilde{Q}_{\hat{\mu}^-} (\tfrac{1}{2}{\bf t}_B)),
\eea
and
\be
i^{m^-(\mu)} = 
(-1)^{\tfrac{1}{2}m^-(\mu)} = 
(-1)^{\tfrac{1}{2}\hat{m}^-(\mu)} .
\ee
If $m^+(\mu)$ and  $m^-(\mu)$ are both odd,  by Lemma \ref{factorization_lemma} and eq.~(\ref{Q-fermi}) we have
\bea
&\&\langle 0 | \ \phi^+_{\mu^+_1} ({\tb}_B) \cdots 
\phi^+_{\mu^+_{m^+(\mu)}}({\tb}_B) \phi^-_{\mu^-_1}({\tb}_B) \cdots \phi^-_{\mu^-_{m^-(\mu)}}({\tb}_B)|0 \rangle \cr
&\& = 2i \langle 0 | \phi^+_{\mu^+_1}({\tb}_B)  \cdots 
\phi^+_{\mu^+_{m^+(\mu)}}({\tb}_B)   \phi^-_0({\tb}_B) |0\rangle
\langle 0 | \phi^-_{\mu^-_1}({\tb}_B) \cdots \phi^-_{\mu^-_{m^-(\mu)}}({\tb}_B) \phi^-_0({\tb}_B) |0 \rangle \cr
&\& =2^{-r}i \tilde{Q}_{\hat{\mu}^+} (\tfrac{1}{2}{\bf t}_B) \tilde{Q}_{\hat{\mu}^-} (\tfrac{1}{2}{\bf t}_B)),
\eea
and
\be
i^{m^-(\mu)+1} = (-1)^{\tfrac{1}{2}m^-(\mu)+1} = (-1)^{\tfrac{1}{2}\hat{m}^-(\mu)} .
\ee
In both cases, substituting these in (\ref{exp_phi_alpha_beta}) gives (\ref{s_lambda_Q_bilinear_bis}).
\end{proof}
\begin{remark}
 Note  that half the $2^{2(r-s)}$ terms in  the sum (\ref{s_lambda_Q_bilinear_bis})
  are the same as the other half (under the interchange \hbox{$(\mu^+, \mu^-)\leftrightarrow (\mu^-, \mu^+)$}), 
  leaving only $2^{2r-2s-1}$ distinct terms (except for the case $s=r$, where there is just one).
\end{remark}

%%%%%%%%%%%%%%%%  Subsubsection 5.1 Examples %%%%%%%%%%%%%
\subsection{Examples.}
\label{examples_doubles}

\begin{example}[cf. Example \ref{example3.1}]
\label{example5.1}
Consider the case of a ``double'' 
\be
D(\alpha)=(\alpha_1, \dots, \alpha_r | \alpha_1-1, \cdots , \alpha_r -1),
\ee
for which $s=r$.
The only binary sequence of type $\epsilon(j_0)$  (i.e.~for which none of the elements $\alpha_m$
corresponds to an upper index $-$) is:
\be
\epsilon (2^{r}-1) = (\underbrace{+, \cdots +}_{r \ \text{terms}}, \underbrace{-, \cdots -}_{r \ \text{terms}}), 
\ee
which gives the canonical polarization 
\be
\mu =\mu(2^r -1) = ((\alpha_1, \dots, \alpha_r), (\alpha_1, \dots, \alpha_r)).
\ee
For this case we have
\be
\sgn(\mu) = +1, \quad \pi(\mu) = r, \quad \hat{m}^-(\mu) =\begin{cases}  r \ \text{ if } \ r \text{ is even }\cr
 r+1 \ \text{ if } \ r \text{ is odd} \end{cases}.
\ee
Therefore
\be
{(-1)^{\tfrac{1}{2}r(r+1) +s} \over 2^{2r -s}}\sgn(\mu)(-1)^{\pi(\mu) + \tfrac{1}{2} \hat{m}^-(\mu)} = {1\over 2^r}, 
\ee
whether $r$ is even or odd, and eq.~(\ref{s_lambda_Q_bilinear_bis}) gives
\be
s_{D(I(\alpha))}(\tb') =
\begin{cases}
 2^{-r}\left(\tilde{Q}_{\alpha}\left(\tfrac12 \tb_B\right)  \right)^2 \ \text{ if } r \ \text{is even},\cr
 2^{-r}\left(\tilde{Q}_{(\alpha, 0)} \left(\tfrac12 \tb_B\right)  \right)^2 \ \text{ if } r \ \text{is odd}, 
 \end{cases} 
 \ee
 in agreement with eq.~(\ref{schur_D_alpha_Q}).
\end{example}

\begin{example}[cf.~Example \ref{example3.2}]
\label{example5.2}
 Consider a hook partition $\lambda=(\alpha_1|\beta_1)$ with $\alpha_1> \beta_1+1$. 
Then $S=  \emptyset $, $(r,s) = (1,0)$. Therefore there are four admissible values for $\mu$:
\bea
(\mu^+(0), (\mu^-(0))&\& =((\alpha_1,\beta_1+1), (\emptyset)), \quad  
(\mu^+(1), \mu^-(1))=((\alpha_1), (\beta_1+1)), \cr
&\& \cr
(\mu^+(2),\mu^-(2))&\&=((\beta_1+1),(\alpha_1)), \quad
(\mu^+(3), \mu^-(3)= ((\emptyset), (\alpha_1,\beta_1+1)).
\eea
Assuming $\alpha_1>\beta_1+1$ to fix the order in the notations and the correct sign factor,
the values of $\sgn(\mu), \pi(\mu)$ and $\hat{m}^-(\mu)$ for each case is:
\bea
\sgn(\mu(0)) &\& =+1, \quad \pi(\mu(0) = 0, \quad \hat{m}^-(\mu(0)) = 0, \cr
\sgn(\mu(1)) &\& =+1, \quad \pi(\mu(1) = 0, \quad \hat{m}^-(\mu(1)) = 2, \cr
\sgn(\mu(2)) &\& =-1, \quad \pi(\mu(2) = 1, \quad \hat{m}^-(\mu(2)) = 2, \cr
\sgn(\mu(3)) &\& =+1, \quad \phi\mu(3) = 1, \quad \hat{m}^-(\mu(3)) = 2, 
\eea
 Substituting in eq.~(\ref{s_lambda_Q_bilinear_bis})  gives
\be
s_{(\alpha_1|\beta_1)}(\tb') =
\frac12 \tilde{Q}_{(\alpha_1,0)}(\tfrac12\tb_B)\tilde{Q}_{(\beta_1+1,0)}(\tfrac12\tb_B)  - \frac12 \tilde{Q}_{(\alpha_1,\beta_1+1)}(\tfrac12\tb_B) \tilde{Q}_\emptyset, 
\ee
which is the same as (\ref{schur_i_j_Q}).
\end{example}

 \begin{example}[cf.~Examples  \ref{example3.2}, \ref{example3.3},  for $r=1,2$]  
\label{example5.3}  
 For  $r\ge1$, we have the following generalization of Example \ref{example5.2}. 
 Assume the Frobenius indices satisfy $\beta_j+1= \alpha_k$ for all pairs except possibly one, 
say  $\alpha_1> \beta_1+1> \alpha_2$, as in Example \ref{example3.3}.
Let 
\be
\lambda=(\alpha_1,\alpha_2,\dots,\alpha_r|\beta_1,\alpha_2-1,\dots,\alpha_r-1),
\ee 
with $\alpha_1 > \beta_1+1> \alpha_2$ if $r > 1$. Then $s=r-1$ and
 there are four possible polarizations $\{(\mu^+(i), \mu^-(i))\}_{i=1, \dots, 4}$.
\bea
(\mu^+(2^r-1), \mu^-(2^r-1)) &\& =((\alpha_1,\alpha_2,\dots,\alpha_r), (\beta_1+1,,\alpha_2,\dots,\alpha_r)), \cr
(\mu^+(2^{2r-1} +2^{r-1} -1), \mu^-(2^{2r-1}+2^{r-1} -1))&\&=((\beta_1+1,\alpha_2,\dots,\alpha_r), (\alpha_1,\alpha_2,\dots,\alpha_r)), \cr
(\mu^+(2^{2r-1}+ 2^r -1), \mu^-(2^{2r-1}+ 2^r -1))&\&=((\alpha_2,\alpha_3,\dots,\alpha_r), (\alpha_1,\beta_1+1,\alpha_2,\dots,\alpha_r)), \cr
(\mu^+(2^{r-1}-1), \mu^-(2^{r-1}-1))&\&=((\alpha_1,\beta_1+1,\alpha_2,\dots,\alpha_r,) (\alpha_2,\dots,\alpha_r)). \cr
&\&
\eea
The corresponding values of $\sgn(\mu), \pi(\,u)$ and $\hat{m}^-(\mu)$ are:
\bea
\sgn(\mu(2^r -1))=1 &\&  , \quad \pi(\mu(2^r -1)) = r-1 , \quad \hat{m}^-(\mu(2^r -1)) = \hat{r} , \cr 
\sgn(\mu(2^{2r-1} +2^{r-1} -1)) = -1&\&  , \quad \pi(\mu(2^{2r-1} +2^{r-1} -1)) = r , \quad \hat{m}^-(\mu(2^{2r-1} +2^{r-1} -1)) =  \hat{r}, \cr 
\sgn(\mu(2^{2r-1}+ 2^r -1)) =(-1)^{r-1}&\&  , \quad \pi(\mu(2^{2r-1}+ 2^r -1)) = r, \quad \hat{m}^-(\mu(2^{2r-1}+ 2^r -1))) = \widehat{r+1}, \cr 
\sgn(\mu(2^{r-1}-1))= (-1)^{r-1} &\&  , \quad \pi(\mu(2^{r-1}-1)) = r-1, \quad \hat{m}^-(\mu(2^{r-1}-1)) = \widehat{r-1}. \cr
&\&
\eea
Substituting these in eq.~(\ref{s_lambda_Q_bilinear_bis})  gives
\bea
s_{(\alpha_1, \dots, \alpha_r |\beta_1,\alpha_2-1,\dots,\alpha_r-1)}(\tb')&\&= 
\tfrac{1}{2^r} \big(\tilde{Q}_{\hat{\mu}^+(2^r-1)}(\tfrac12\tb_B)\tilde{Q}_{\hat{\mu}^-(2^r-1)}(\tfrac12\tb_B)\cr
&\& \cr
 &\&\quad -  \tilde{Q}_{\hat{\mu}^+(2^{r-1}-1)}(\tfrac12\tb_B)\tilde{Q}_{\hat{\mu}^-(2^{r-1}-1)}(\tfrac12\tb_B)\big), 
\eea
in agreement with (\ref{schur_i_j_Q}) for $r=1$ and (\ref{s_r2_s1}) for $r=2$.
\end{example}
\bigskip
\bigskip
%%%%%%%%%%%%%%%%%%%%% Acknowledgements %%%%%%%%%%%%%%%%%
\noindent 
\small{ {\it Acknowledgements.} The authors would like to thank Johan van de Leur  and Ferenc Balogh for helpful discussions.
The work of J.H. was supported by the Natural Sciences and Engineering Research Council of Canada (NSERC);
that of A. Yu. O.  by the P.P. Shirshov Institute of Oceanology RAS state assignment No 0128-2021-0002.

\end{document}